\documentclass[12pt,reqno]{amsart}
\usepackage{amssymb,amsmath,bm}
\usepackage{graphicx}
\usepackage{caption}
\usepackage{color}
\usepackage[usenames,dvipsnames]{xcolor}
\usepackage{tikz}
\usepackage[colorinlistoftodos]{todonotes}
\def\eps{\varepsilon}
\def\d{{\rm d}}

\def\R {\mathbb{R}}

\def\e {\varepsilon}

\def\z{\zeta}
\def\O{\Omega}

\def\bp{p}

\def\bbQ{\mathbb{Q}}
\def\bbY{\mathbb{Y}}

\def\calM{\mathcal M}

\def\QQ {\mathcal{Q}}

\newcommand{\FF}{{\mathcal F}}

\newcommand{\HH}{{\mathcal H}}

\newtheorem{proposition}{Proposition}[section]
\newtheorem{theorem}[proposition]{Theorem}
\newtheorem{corollary}[proposition]{Corollary}
\newtheorem{lemma}[proposition]{Lemma}
\theoremstyle{definition}
\newtheorem{definition}[proposition]{Definition}

\numberwithin{equation}{section}
\newcommand{\epsi}{\varepsilon}

\newcommand{\Nz}{{\mathbb N}}
\newcommand{\Rz}{{\mathbb R}}

\newcommand{\dx}{\,\text{\rm d}x}

\newcommand{\per}{\mathrm{Per}}

\newcommand{\cof}{\mathrm{cof\,}}

\newcommand{\md}{{\rm d}}
\newenvironment{proofreverse}{\removelastskip\par\medskip   
\noindent{\em Proof of {\rm Theorem \ref{reverse}}.}
\rm}{\penalty-20\null\hfill$\square$\par\medbreak} 

\newenvironment{proofth2}{\removelastskip\par\medskip   
\noindent{\em Proof of {\rm Theorem \ref{th2}}.}
\rm}{\penalty-20\null\hfill$\square$\par\medbreak} 

\def \no#1#2#3 {{\bf #1} (#3), #2.}
\def \eds#1#2#3 {#1, #2, #3.}
\title[A phase-field approach to  Eulerian   interfacial energies]
{A phase-field approach to \\  Eulerian  interfacial energies}

\author{Diego Grandi}
\address[Diego Grandi]{Dipartimento di Matematica e Informatica, 
  Universit\`a  degli Studi di Ferrara,
Via Machiavelli 30, 44121 - Ferrara, Italy.}
\email{diego.grandi@unife.it}

\author{Martin Kru\v{z}\'ik}
\address[Martin Kru\v{z}\'ik]{
Czech Academy of Sciences, Institute of Information Theory and Automation,
Pod vod\' arenskou ve\v z\' \i\ 4, 182 08, Prague 8, Czech Republic and Faculty of Civil Engineering, Czech Technical University, Th\'{a}kurova 7, 166 29, Prague 6, Czech Republic.}
\email{kruzik@utia.cas.cz}

\author{Edoardo Mainini}
\address[Edoardo Mainini]{Dipartimento di Ingegneria meccanica, energetica, gestionale e dei trasporti, 
  Universit\`a  degli studi di Genova, Via all'Opera Pia, 15 - 16145 Genova Italy.}
\email{mainini@dime.unige.it}

\author{Ulisse Stefanelli}
\address[Ulisse Stefanelli]{Faculty of Mathematics, University of Vienna, 
Oskar-Morgenstern-Platz 1, 1090 Wien, Austria  and  Istituto di Matematica
Applicata e Tecnologie Informatiche \textit{{E. Magenes}}, v. Ferrata 1, 27100
Pavia, Italy.}
\email{ulisse.stefanelli@univie.ac.at}

\keywords{}
\begin{document}
\begin{abstract}
 We analyze a phase-field approximation of a sharp-interface model
for two-phase materials proposed by M.~\v{S}ilhav\'{y}
\cite{S2,S3}. The distinguishing trait of the model resides in the
fact that the interfacial term is Eulerian in nature, for it is
defined on
the deformed configuration. We discuss a functional frame allowing for
existence of phase-field minimizers and $\Gamma$-convergence to the
sharp-interface limit. As a by-product, we provide additional detail
on the admissible sharp-interface configurations with respect to the
analysis in
\cite{S2,S3}.

\end{abstract}
\maketitle
%
%
%
\section{Introduction}

This paper addresses the equilibrium of a two-phase elastic medium,
whose stored energy takes the form
\begin{align}
&\FF_0(y,\z)=  \FF^{\rm bulk}(y,\z)+ \FF_0^{\rm int}(y,\z)
\nonumber\\
&:=\int_\Omega\Big( (\z\circ y)\,
   W_1(\nabla y)+(1-\z\circ y)\,W_{ 0 }(\nabla y)\Big)\,\dx +
   \gamma\, \per\big(\{\z=1\},y(\Omega)\big). \label{effezero} 
 \end{align}
Here, $y:\Omega \to \R^3$ stands for the deformation of the medium with
respect to its {\it reference} configuration $\Omega\subset \R^3$ 
and  $W_0$,
$W_1$ are the elastic energy densities of the two pure phases \cite{silhavy-book}. The Eulerian
phase indicator $\z: y(\O) \to \{0,1\}$ is defined on the {\it
  deformed} configuration $y(\O)$ instead. Note that solely pure
phases are allowed. The stored energy of the
medium includes an elastic bulk part $\FF^{\rm bulk}(y,\z)$, consisting of an
integral on the reference configuration, and an interface contribution
$\FF_0^{\rm int}(y,\z)$, featuring the perimeter of the phase
$\{\z=1\}$ in $y(\O)$,  where $\gamma>0$ is a surface-tension coefficient.  With
respect to classical hyperelastic theory, the novelty in
\eqref{effezero} is that the 
interface is measured {\it in the deformed configuration}, giving rise
to a variational model of mixed Lagrangian-Eulerian type.

The choice of the elastic energy $\FF_0$ is inspired by the notion of \emph{interface polyconvex energy}, introduced  by
M.~\v{S}ilhav\'y in the series of contributions \cite{S1,S2,S3}. 
The explicit form in \eqref{effezero} is  in fact just a first
example in the wider class    considerer therein,  where the
general interfacial term  reads 
\begin{align}\label{ex-silhavy0}\int_{\partial E\setminus\partial\O}\Psi(n, \nabla_S y\times
n,(\cof\nabla_S y)n)\,\d S.
\end{align}
 Here, $ \d S $ is the infinitesimal area element (in the
 reference  configuration), and $\Psi:\R^{15}\to\R$ is a positively
 $1$-homogeneous convex function  depending on  the normal  $n$ 
 to the
interface, on the surface gradient $\nabla_S y$
of the deformation, and on the cofactor of the surface gradient. 
More precisely,  $\Psi=\Psi( n,\mathbb{F}\times
n,\cof \mathbb{F}n)$, where $\mathbb{F}\in\R^{3\times 3}$ is a
placeholder of the surface gradient of the deformation  and
$\mathbb{F}\times n:\R^3\to \R^3$ is defined for all $a\in \R^3$ as $(\mathbb{F}\times
n)a:=\mathbb{F}(n\times a)$.  Note that
$\mathbb{F}n=0$, because $n$ inevitably lives  in the kernel of
$\mathbb{F}$. 
A rigorous definition would ask to cope with
 the possible  
nonsmoothness of $y$, the existence of the surface gradient  $\nabla_S
y$,
and also whether $n$ does exist  at the phase interface,  which in
turn  
relates with the regularity of  phase $1$ in the reference
configuration, for $E=y^{-1}(\{\z =1\})$  in \eqref{ex-silhavy0}.  The  specific interfacial term in
\eqref{effezero}   corresponds to  the choice \cite[Ex.~5.7]{S3}
 \begin{align}\label{ex-silhavy}
\tilde \Psi( n,\mathbb{F}\times n,\cof \mathbb{F}n):=\gamma|\cof \mathbb{F}n|.
\end{align}
 Indeed, it  is a standard matter to check that  $(\cof\nabla_S y)n=(\cof\nabla
y)n$. Then, a formal application of the change-of-variables  formula
for  surface integrals \cite{ciarlet} gives
\begin{align}\label{formal-per}
\int_{\partial E \setminus \partial \Omega}\tilde
  \Psi( n,\mathbb{F}\times n,\cof \mathbb{F}n) \, \md S =
  \gamma 
\int_{\partial E \setminus \partial \Omega}|(\cof \nabla y)n|\,\md S= 
\gamma \int_{\partial y(E)\setminus\partial y(\O)}\,\md S^y.
\end{align}
 As  $\md S^y$ is the
infinitesimal area element in the deformed configuration
$y(\O)$,    we have checked that,  along with choice
\eqref{ex-silhavy},   the interfacial
energy term measures indeed the surface of the
interface in the deformed configuration.  This  is consistent with the definition of $\FF^{\rm int}_0$ from \eqref{effezero}. 

Our main results are the existence of minimizers of $\FF_0$ (Theorem
\ref{th1}) and the viability of a phase-field approach (Theorem
\ref{th2}) to such  
 sharp-interface model via the diffuse-interface energies for $\e>0$
\begin{align}
&\FF_\e(y,\z)=  \FF^{\rm bulk}(y,\z)+\FF_\e^{\rm int}(y,\z):=\FF^{\rm bulk}(y,\z) +  \int_{y(\O)}\Big( \frac\eps2| \nabla 
  \z|^2+\frac1\eps \Phi(\z)\Big)\,\d\xi.\label{effee}
 \end{align}
Note that the diffuse-interface term  $\FF_\e^{\rm int}(y,\z)$ is
still Eulerian,  but the phase indicator $\z$ takes now values in the
interval $[0,1]$. Here and throughout the paper, $\xi$ stands for the
variable in the deformed configuration $y(\Omega)$.  The function
$\Phi$  in \eqref{effee}  is a classical double-well potential with
minima at $0$ and $1$,  and $\int_0^1\sqrt{2\Phi(s)}\, \d  s=\gamma$.   By checking the $\Gamma$-convergence of $\FF_\e$ to
$\FF_0$ we essentially deliver a version of the Modica-Mortola Theorem~\cite{MM} {\it in the deformed configuration}. Instrumental to this is
the discussion of the interplay of deformations and perimeters in
deformed configurations, which constitutes the main technical
contribution of our paper  (Theorem \ref{reverse}).

Let us mention that variational formulations
featuring both Lagrangian and Eulerian terms are currently attracting
increasing attention.
A prominent case is that of 
 magnetoelastic materials \cite{desimone.james}, where Lagrangian
 mechanical terms and Eulerian magnetic effects combine \cite{barchiesi.henao.mora-corral,bielsky.gambin,kruzik.stefanelli.zeman,rybka.luskin}.
Mixed Lagrangian-Eulerian formulations arise in the modeling of
nematic polymers  \cite{barchiesi.desimone,barchiesi.henao.mora-corral}, where the Eulerian variable is the nematic director
orientation, and
in piezoelectrics \cite{rosato.miehe}, involving the Eulerian
polarization instead. An interplay of Lagrangian and Eulerian
effects occurs already in case of space dependent forcings, like in
the variable-gravity case \cite{fitzpatrick}, as well as in specific 
finite-plasticity settings \cite{stefanelli}, where elastic and plastic deformations are
composed. Most notably,  such mixed formulations arise naturally in the study of fluid-structure interaction, where the deformed body defines the (complement of the) fluid domain \cite{richter}.

 The plan of the paper is as follows. We  present in detail our  assumptions on
the ingredients of the models in Section
\ref{mainsection}. In particular, we specify the class of admissible
deformations and state a  characterization of
sets of finite perimeter with respect to deformed configurations
(Theorem \ref{reverse}). Subsection \ref{main} contains the statements
of our main existence and approximation results. These are put in
relation with the former theory by M.~\v{S}ilhav\'y in Subsection
\ref{sil}.
We check in Section \ref{injectivity} that admissible deformations are
actually homeomorphisms, so that, in particular, the deformed
configuration is well defined. The proof of the Characterization
Theorem \ref{reverse} is presented in Section \ref{sil2}, along with a
suite of results on perimeters in deformed configurations. The
existence of minimizers to $\FF_0$ (Theorem \ref{th1}) is proved in
Section \ref{conve}. Eventually, Section \ref{lastsection} proves the $\Gamma$-convergence of the phase-field
diffuse-interface energies $\FF_\e$ to the sharp-interface limit
$\FF_0$ (Theorem \ref{th2}).  

\section{Main results} \label{mainsection}

 We devote this section to specifying the functional frame
(Subsections \ref{finite}-\ref{assu2}) and
stating  our  main results (Subsection \ref{main}). The relation of our results
with the former existence theory by M.~\v{S}ilhav\'y is also discussed
(Subsection \ref{sil}).

We first introduce some basic notation. We denote
by $B(a,\eps):=\{z \in \R^n \ | \ |z- a|<\eps\}$ the open ball of
radius $\eps>0$ centered at $a\in \mathbb{R}^n$.
If $ \O\subset\R^n$ is an open set, $C^m(\O;\R^k)$ denotes the space
of continuous maps on $\O$ with values in $\R^k$ that admit continuous
derivatives up to the order $m\ge 0$.  $C^m_{\rm  c}(\O;\R^k)$ is the subspace of compactly supported  maps.
    For $p\in[1,+\infty)$, $W^{1,p}(\O;\R^k)$  denotes the standard Sobolev space, and $W^{1,p}_{\rm loc}(\O;\R^k)$ denotes its local counterpart.   
The
space of finite vector Radon measures on $\O$ with values in $\R^k$ is
denoted by $\mathcal{M}(\O,\R^k)$ and it is normed by the total
variation $|\cdot|(\O)$. $\mathcal{M}_{\rm  loc}(\O;\R^k)$ denotes the space of locally finite vector Radon measures. Furthermore, $BV(\O;\R^k)$ stands for the space of maps with bounded variation. See
e.g.~\cite{AFP} for references. 
With  slight abuse of  notation, we occasionally replace $\R^k$ in the target space by a set.
 For a measurable set $E \subset\Omega$, we
denote the $n$-dimensional Lebesgue measure by $|E|$ and the
$m$-dimensional Hausdorff measure by $\HH^m(E)$.  By   $\chi_E$ we
denote the characteristic function of $E$.  The perimeter of   $E$
in $\O$ is classically defined as 
\cite[Def. 3.35]{AFP}  
$$\per(E,\O):=\sup\left\{\int_{E}\textrm{div}\varphi\, \d x\ | \ \varphi\in C^\infty_{\rm c}(\O;\R^n),\;\|\varphi\|_\infty\leq 
1\right\}.$$

%
%

 Given  $y: E \to \Rz^3$, we will use the notation $E^y := y(E)$.

\subsection{ Finite distorsion and finite perimeter}\label{finite}

 Let us start by defining the function classes that we are going to
be dealing with. 

\begin{definition}[Finite distorsion]\label{finitedistorsion} 
 Let $\O\subset\mathbb{R}^n$  for $n\ge 2$  be an open set. 
A Sobolev map $y\in W^{1,1}_{\rm loc}(\Omega;\R^{ n})$ with $\det\nabla y\ge 0$ almost everywhere in $\Omega$ is said to be of {\it finite distorsion} if $\det\nabla y\in L^1_{\rm loc}(\Omega)$ and  there is  a function $K:\O\to[1,+\infty]$ with $K<+\infty$ almost everywhere in $\O$ such that $|\nabla y|^n\le K\det\nabla y$. 
For a mapping $y$ of finite distorsion, the {\it (optimal) distorsion function} $K_y:\Omega\to\mathbb{R}$ is defined as
\[
K_y:=\left\{\begin{array}{ll}{|\nabla y|^{ n}}/{\det\nabla y}\quad&\mbox{if $\det\nabla y\neq0$}\\
1\quad&\mbox{if $\det\nabla y=0$}.\end{array}\right.
\]
\end{definition}

   The
relation of our theory to the former one by  M.~\v{S}ilhav\'y is
encoded in the following characterization result for sets of finite
perimeters in the actual configuration. Although it will be later
applied just for $n=3$, we state the characterization here for general
dimension, for we believe that it could be of  independent interest. 
 
 \begin{theorem}[Characterization of sets of  finite perimeter]\label{reverse} Let $\Omega\subset\mathbb{R}^n$ be an  open set, $n\ge 2$.
Suppose that $E\subset\Omega$ is a measurable set and that $y\in
W^{1,n}_{\rm loc}(\Omega;\R^n)$ is a homeomorphism of finite distorsion.
Then $\per(E^y,\Omega^y)<\infty$ if and only if  there exists a finite Radon measure $\bp_{y,E}\in\mathcal{M}(\Omega;\R^n)$ such that there holds
\begin{equation}\label{eq:p_yE-1}
 \int_{ E}\cof(\nabla 
y):\nabla\psi \, \d x=\int_{\O}\psi\cdot \d \bp_{y,E}\qquad \forall \psi\in C^\infty_{\rm c}(\O;\R^n).
\end{equation}
 In this case, $\per(E^y,\Omega^y) = |\bp_{y,E}|(\O)$. 
\end{theorem}

 A proof of this  characterization is  provided in  Section
\ref{sil2}. 

In the following, we call $\bp_{y,E}$ a {\it \v{S}ilhav\'y
  measure} if it is a finite Radon measure and 
it fulfills  \eqref{eq:p_yE-1}
for some $y$ and $E$ within the assumption frame of Theorem \ref{reverse}. This
naming is hinting to the relevance that such measures enjoy within the
theory  by M.~ \v{S}ilhav\'y  \cite{S1,S2}, see Subsection \ref{sil} below. Theorem
\ref{reverse} proves in particular that, given an admissible
deformation $y$,  \v{S}ilhav\'y measures correspond one-to-one  to sets of
finite perimeter in the deformed configuration $\O^y$.

Notice in particular that, by taking $y$  to be  the identity
map on $\O$,   Theorem \ref{reverse} reduces  to the 
classical  characterization of sets $E$ of finite perimeter in
$\O$  \cite[Thm. 3.36]{AFP}, namely those sets such that 
there exists a finite measure $\bp_E$  with 
\[
 \int_{ E}\mathrm{div}\psi\, \d x=\int_{\O}\psi\cdot \d \bp_{E}\qquad \forall \psi\in C^\infty_{\rm c}(\O;\R^n).
\]

\subsection{ Admissible states}\label{assu}

From now on let the open, bounded, and Lipschitz domain $\O\subset \R^3$ indicate the
reference configuration. The body undergoes a deformation
 $y:\O\to\R^3$, which is assumed to be a Sobolev mapping of finite
 distorsion. We will  in fact ask  that $y$ is orientation-preserving, i.e., $\det\nabla y>0$ almost everywhere in $\O$.
It is well known that positivity of $\det\nabla y$ ensures only 
the  local injectivity of $y$ \cite{ciarlet}. However, it is shown
by Ciarlet and Ne\v{c}as  \cite{ciarlet-necas} that  if $y\in
W^{1,p}(\O;\R^3)$ for some $p>3$, $\det\nabla y>0$ almost everywhere,
and additionally the so-called {\it Ciarlet-Ne\v{c}as condition}
\begin{align}\label{inj}
\int_\O\det\nabla y(x)\,\md x\le | \O^y |\ 
\end{align}
 holds, then almost every point in $  \O^y $ has only one
 preimage. Under such assumptions, as we will  thoroughly discuss
  in Section \ref{injectivity},  everywhere injectivity (so that
 the deformation is a homeomorphism) can be further enforced by
 requiring that  the distorsion function  $K_y$  is in $L^q(\Omega)$
 for some $q>2$.

 Therefore, we define the set of {\it admissible deformations} as  
\begin{equation} \label{eq:Y}
\begin{aligned}
\bbY:=\left\{y \in W^{1,p}(\O;\mathbb R^3)\ | \ \det\nabla y>0 \text{ a.e.}\, , 
 \int_\O\det\nabla y(x)\,\md x\le |y(\O)|,\, 
  K_y\in L^q(\Omega)
 \right\}
\end{aligned}
\end{equation}
where $p>3$ and  $q>2$ are fixed.  We shall check in Section
\ref{injectivity} that admissible deformations are homeomorphisms, see
Theorem \ref{Thm:MIE}. In particular, the  deformed configuration
$\O^y$    is an open
set. 
 
 We consider a material with two different phases (e.g.,~two
 martensitic variants of a shape memory alloy)  which we indicate
 with the subscripts $0$ and $1$.  To indicate 
 the portion $E\subset \Omega$ of the reference configuration where
 one finds phase $1$, one defines $z:\Omega \to \{0,1\}$ and $\zeta:
 \Omega^y \to \{0,1\}$ to be the
 characteristic functions of $E$ and $ E^y $, respectively. In
 particular, we have that $z=\zeta \circ y$.

The set  of {\it admissible states} $(y,\zeta)$  is defined as
$$
\mathbb{Q}:= 
\{(y,\zeta)\ | \ y\in \mathbb{Y}, \ \zeta\in BV(\Omega^y;\{0,1\})\}.
$$ 
Similarly, we define the set of admissible  states  for the phase-field
approximation  as 
\[
\overline{\mathbb{Q}}:= 
\{(y,\zeta) \ | \ y\in \mathbb{Y}, \ {\zeta}\in BV(\Omega^y;[0,1])\}.
\]

\subsection{Assumptions on the bulk energy} \label{assu2}
We assume that   $W_0$ and $W_1$  are {\it polyconvex} \cite{Ball-1977}, i.e., for $F \in \R^{3 \times 3}$
\begin{align} \label{W-ass1} %
  W_i(F ):=
  \begin{cases}
    h_i(F ,\cof F ,\det F ) & \mbox{ if } \det F >0, \\
    \infty \mbox{ otherwise}
  \end{cases}
\end{align}
for some convex functions $h_i:\R^{19} \to \R$, $i= 0, 1$. In
addition, we assume $W_i$ to be coercive, frame-indifferent, and
unbounded as $\det F \to 0+$. More precisely,  for $i=0,1$, we
assume that there exist  $C>0$  such that 
\begin{align}
  & W_i(F )\ge C\left(|F |^p+\frac{|F |^{3q}}{(\det F )^{q}}  - 1\right)
  &&\forall F  \in \R^{3 \times 3}, \, p>3, \ q>2  , \label{W-ass2} \\
  & W_i(RF)= W_i(F ) && \forall R\in{\rm SO}(3), F  \in \R^{3\times 3} \ , \label{W-ass3}\\
  & W_i(F ) \to \infty \ \ \text{as}  \ \  \det F \to 0_+  \label{W-ass4}
\end{align}
 where ${\rm SO}(3)$ is the special orthogonal group ${\rm
  SO}(3)=\{R\in  \R^{3\times 3} \ | \ R R^T =  I, \ \det
R=1\}$. 
The third term on the right-hand side of \eqref{W-ass2} ensures that
deformation  gradients  $F=\nabla y$ with finite energy  will
have a $q$-integrable distorsion function  $F\mapsto
|F|^3/\det F$. Notice that $F\mapsto|F|^3/\det F$ is polyconvex on the
set of matrices with positive determinant.

Eventually, we specify boundary conditions by imposing admissible deformations to match a given
deformation $y_0$ at the boundary $\partial \Omega$. To this aim, we
assume that
\begin{equation}
\exists (y_0,\z_0)\in\bbQ \ \ \text{with} \ \ \mathcal{F}_{0}(y_0,\z_0)<\infty\label{W-ass5}
\end{equation}
and define
$$\mathbb{Q}_{y_0}:=\{(y,\zeta)\in\mathbb{Q}\ | \
y=y_0\;\text{on}\; \partial\O\}.$$
Analogously, we consider 
$$\overline{\mathbb{Q}}_{y_0}:=\{(y,\zeta)\in \overline{\mathbb{Q}}\ | \
y=y_0\;\text{on}\; \partial\O\}.$$


%

%
\subsection{Main results}\label{main}
 We are now in the position of stating the main results of the
paper, which concern existence for the sharp-interface minimization
problem and convergence of the phase-field approximation.  

\begin{theorem}[Existence of minimizers]\label{th1}
 Under  assumptions \eqref{W-ass1}-\eqref{W-ass5}  
the functional $\mathcal{F}_0$ admits a minimizer on
$\mathbb{Q}_{y_0}$. 
\end{theorem} 

A proof of this statement is in Section \ref{conve}.

 Our second main result  delivers a  Modica-Mortola-type 
approximation  via the  functionals $\mathcal{F}_\eps$ from
\eqref{effee}, corresponding  indeed   to  diffuse-interface models. 
Under the  additional  assumption that the current
configuration  $\O^y$  is a Lipschitz domain (which is  not 
necessarily  true for general    $y\in W^{1,p}(\O;\R^3)$) we have the following 

\begin{theorem}[Phase-field approximation]\label{th2}
 Under assumptions  
 \eqref{W-ass1}-\eqref{W-ass5}, 
 for any $\eps>0$ the  functional $\mathcal{F}_\eps$ admits  a minimizer
  on   $\overline\bbQ_{y_0}$. 
   If $\Omega^{y_0}$ is a Lipschitz domain and $\eps_k \to 0 $, then,  for every sequence
 $(y_k,\zeta_k)$ of minimizers of $\mathcal{F}_{\eps_k}$  on 
 $\overline\bbQ_{y_0}$, there exists $(y, \z)\in \bbQ_{y_0}$ such
 that, up to  not relabeled subsequences, 
 \begin{itemize}
 \item[i)] $y_k\to y$ weakly in $W^{1,p}(\O;\R^3)$, 
  $|\Omega^{y_k}\Delta\O^y|\to 0$, and  $\|\zeta_k-{\zeta}\|_{L^1(O^k)}\to0 $ as $k\to\infty$, where $ O^k:=\O^{y_k}\cap \O^y$.
 \item[ii)] $(y,\z)$ minimizes $\mathcal{F}_0$ on $\bbQ_{y_0}$.
 \end{itemize}
 \end{theorem}

%
\subsection{Relation with \v Silhav\'y's theory}\label{sil}
%
%
 Before moving on, let us comment on our results in light of the
theory by M.~\v{S}ilhav\'{y}  \cite{S1,S3}. To this end, we need to
clarify the definition of the general interfacial-energy term in
\eqref{ex-silhavy0},
 which requires  introducing some  measure theoretic setting. We recall that the \emph{reduced boundary} of a finite perimeter set $E$ in $\Omega$ is defined as the set of points $x$ of $\Omega$ such that
$x\in{\rm supp}\, |\nabla\chi_E|$ and  such that the limit
$
n_E(x):=\lim_{\eps\to 0}\frac{-\nabla\chi_E(B(x,\eps))}{|\nabla\chi_E|(B(x,\eps))}$
exists and satisfies $|n_E(x)|=1$ (see \cite[Def. 3.54]{AFP}). We say that $n_E$ is  the \emph{outer measure-theoretic unit normal} to $E$.
We let 
$$\QQ:=\{(y,z)\ | \ y\in W^{1,p}(\Omega),\, \det\nabla y>0\, \text{ a.e. in } \Omega,\;   z\in BV(\O;\{0,1\})\}.$$ 
  For any pair $(y,z)\in \QQ$,
  let $E:=\{z=1\}$,  let $S$ denote the reduced boundary of the finite perimeter set $E$ in $\O$, and let ${n}_E$ denote the corresponding outer measure-theoretic  unit normal.
 Following \cite[Def. 3.1]{S2}, we denote  by $\QQ_0 \subset
\QQ$ the set of all pairs $(y,z)\in \QQ$ for which there exists a
finite Radon measure $m_{y,E}:=({a}_{y,E},{h}_{y,E},{p}_{y,E})\in \calM(\O;\R^{15})$ such that ${a}_{y,E} := {n}_E \HH^2_{|S}$ and such that there hold \eqref{eq:p_yE-1} and
  \begin{align} \label{form-1} %
  \int_{ E} \nabla y\ (\nabla \times \psi) \,\md x= \int_\O \psi \, \d{h}_{y,E} \qquad \forall \psi\in C^\infty_{\rm c}(\O;\R^3).
\end{align} %
	Consider a  positively $1$-homogeneous   convex  function $\Psi:\R^{15}\to\R$  such that 
	\begin{align}\label{g-ass4}
\Psi(A)\ge C|A| \quad\mbox{for some $C>0$ and all $A\in\R^{15}$.}
\end{align}
    If $|{m}_{y,E}|$ denotes the total variation of  ${m}_{y,E}$, the interfacial energy is then defined as
  \begin{align} \label{stat-interface} \FF^{\rm int}_{\text{\v{S}ilhav\'y}}(y,z):=\begin{cases}
      \displaystyle \int_\O \Psi\left(\frac{\md {m}_{y,E}}{\md |{m}_{y,E}|}\right)\ \md|{m}_{y,E}|\ %
      &\text{for $(y,z) \in  \QQ_0$,} \\
      +\infty & \text{otherwise.}
			\end{cases}
  \end{align}

  On the other hand, the bulk energy  in the reference configuration  
  is defined as
$$\widetilde \FF^{\rm bulk}(y,z) := \int_\Omega\Big( z\,
   W_1(\nabla y)+(1-z)\,W_{ 0 }(\nabla y)\Big)\,\dx$$
where 
$W_i$  are  assumed to satisfy \eqref{W-ass1}, \eqref{W-ass3}-\eqref{W-ass4}, and $W_i(F)\ge C|F|^p$ for $i=0,2$ and some $p>3$.
Under such assumptions on $W_i$ and \eqref{g-ass4},  \v{S}ilhav\'{y}
proves that $\widetilde\FF^{\rm bulk}(y,z)+\FF^{\rm int}_{\text{\v{S}ilhav\'y}}(y,z) $ admits a minimizer  on
$\{(y,z)\in\mathcal{Q}_0\ | \ y=y_0 \text{ on $\partial\O$}\}$,  see
 \cite[Thm.~3.3]{S2} and \cite[Thm.~1.2]{S3}. 
 Our Characterization  Theorem \ref{reverse} shows in
particular that, under the further assumption of $y$ being a
homeomorphism, the perimeter of the image set $ E^y=\{z=1\}^y $
is finite in $ \O^y $.  More specifically,  Theorem
\ref{reverse} provides a characterization of those deformations that
admit a  \v{S}ilhav\'y measure ${p}_{y,E}\in
\mathcal{M}(\Omega;\mathbb{R}^3)$. 

 The existence result of  Theorem \ref{th1} refers to the specific
case \eqref{ex-silhavy} within the larger class
\eqref{ex-silhavy0}. As such, the global coercivity assumption
\eqref{g-ass4} is not required.

%
\section{ Admissible deformations are homeomorphisms}\label{injectivity}
%
 The aim of this section is to check that the continuous
representative of the class of the admissible deformation $y\in
\mathbb{Y}$ \eqref{eq:Y}
is injective, hence a homeomorphism between $\O$ and $\O^y $, see
Theorem \ref{Thm:MIE} below. We break down the argument into Lemmas,
which we believe to be of  an independent interest. Let us   
start with   a definition. 

 \begin{definition}[almost-everywhere injectivity]\label{def:injectivity}
We say that $y:\O\to\R^3$ is {\it almost-everywhere injective } 
  if there  exists  $\omega\subset\O$ such that $|\omega|=0$ and 
$y(x_1)\ne y(x_2)$  for every  $x_1,x_2\in\O\setminus\omega$
satisfying $x_1\ne x_2$. 
\end{definition}

 Given  $y:\O\to\R^3$, $\xi\in\R^3$, and a subset $\omega\subset\O$,
we define the   {\it Banach indicatrix} $N(\xi,y,\omega)$  by 
\begin{align}\label{indicatrix}
N(\xi,y,\omega):=\#\{x\in\omega\ | \  y(x)=\xi\}\ ,
\end{align}
where the right-hand-side denotes the cardinality (i.e., the number of
elements) of the set. The map
$y:\O\to\R^3$ is said to satisfy {\it Lusin's condition $N$} if  it
maps negligible sets to negligible sets, namely $| \omega^y |=0$ for
all  $\omega\subset\O$ such that $|\omega|=0$.  Moreover, it satisfies  {\it Lusin's condition
  $N^{-1}$} if the preimage of any negligible set is negligible,
namely $|y^{-1}(\omega)|=0$ for all  $\omega\subset  \Omega^y $ such
that $|\omega|=0$.

 Any continuous  map  $y\in W^{1,p}(\Omega;\mathbb{R}^3)$,  $p>3$,
satisfies the Lusin's condition $N$  \cite[Theorem 4.2]{HK}. This implies the validity of the area formula with equality \cite[Theorem A.35]{HK}.  If
in addition $\det\nabla y>0$ almost everywhere in $\O$, $y$ satisfies
Lusin's condition $N^{-1}$ as well \cite[Thm.~8.3,
Lem.~8.3-8.4]{bojarski-iwaniec}. This in particular implies that the
continuous representative of $y\in \mathbb{Y}$ fulfils both Lusin's
$N$ and $N^{-1}$ condition.

 Let us present a first result on almost-everywhere
injectivity, see \cite[Prop. 3.2]{Giacomini-Ponsiglione} for a similar argument. 

\begin{lemma}[Ciarlet-Ne\v{c}as implies almost-everywhere
injectivity]\label{lemma:cn}  Let   $y\in W^{1,p}(\O;\R^3)$ be continuous,  $p>3$, and
$\det\nabla y>0$ almost everywhere in $\O$.   If the
Ciarlet-Ne\v{c}as condition \eqref{inj}  holds,  then $y$ is
almost-everywhere injective in the sense of  Definition ~\emph{\ref{def:injectivity}}.
\end{lemma}
\begin{proof}
 The map $y$ satisfies  Lusin's  condition $N$. Hence, the area formula
 holds with equality.  The Ciarlet-Ne\v{c}as condition \eqref{inj}
 implies that   
 $$
 |\O^y |\le\int_{y(\O)} N(\xi,y,\O)\,\d \xi=\int_\O\det\nabla y\, \d x\le |\O^y |,
 $$   which entails  $N(\xi,y,\O)=1$ for almost every $\xi\in
\O^y $.  The set  $\omega := \{ \xi \in y(\Omega) \  |  \ N( \xi,y,\Omega)
 >1\} $ is hence negligible.    Since by \cite[Thm.~8.3,
Lem.~8.3-8.4]{bojarski-iwaniec}
 $y$ satisfies Lusin's condition $N^{-1}$,  we get  that
 $|\{x\in\O \ | \  y(x)\in\omega\}|=0$  as well, which corresponds
 to the statement. 
\end{proof}
 Maps that are almost-everywhere injective  still include rather
 nonphysical situations,  for   a dense, countable set of
 points could be mapped to  a single   point.  We shall
 hence present a result in the direction of everywhere injectivity. 

\begin{lemma}[a.e. injectivity and openness imply injectivity]\label{lemma:ine}
 Let  $y:\O \to \R^3$ be continuous,
almost-everywhere injective,  open  (maps open sets to open sets),
and   fulfill Lusin's condition N. Then,  $y$ is
everywhere injective in $\O$.  
\end{lemma}
\begin{proof}
Assume  by contradiction that $y$ is not everywhere injective,
 i.e. that there  exist $x_1,\, x_2 \in\O$ with $x_1 \not =
x_2$  such that $y(x_1)=y(x_2)=:a$.  The openness of $y$
implies that   $\O^y $ is open.  We can hence find  $\varepsilon>0$ such that $B(a,\varepsilon)\subset \O^y $. Continuity 
 implies that $y^{-1}(B(a,\varepsilon))\subset\O$ is open.  As 
  $x_1,\, x_2\in y^{-1}(B(a,\varepsilon))$ one can find  two open
  disjoint neighborhoods $U,V$ such that $x_1\in U$, $x_2\in V$ and
  $ U^y \cap  V^y \ni a$.
As $ U^y $ and $ V^y $ are both open their intersection is also open and therefore $| U^y \cap  V^y |>0$, i.e. $N(\xi,y,\O)>1$ for every $\xi\in  U^y \cap  V^y $. On the other hand, the pre-image of 
$ U^y \cap  V^y $ must have a positive 
 measure because $y$ satisfies Lusin's condition $N$. This contradicts almost-everywhere injectivity and concludes the proof.
\end{proof}

 Let us now recall a sufficient condition for the openness of a
map.   

\begin{lemma} [$\mbox{\cite[Thm.~3.4]{HK}}$]\label{lemma:ODID}
  Let $y\in W^{1,p}(\O;\R^3)$ for some $p>3$.  Assume that   $K_y\in L^{q}(\O)$ for some $q>2$.  Then $y$ is either constant or  open.
\end{lemma}

 We are finally in the position of stating the main result of this
section. 

\begin{theorem}[Admissible deformations are homeomorphisms]\label{Thm:MIE}
  The continuous representative of $y \in\mathbb{Y}$ is everywhere
 injective on $\O$. 
\end{theorem}

\begin{proof}
Let $y\in\mathbb{Y}$ be the continuous representative of the equivalence class.
Lemma~\ref{lemma:ODID} implies that $y$ is either constant or  open. However, it cannot be constant because it is almost everywhere injective by Lemma \ref{lemma:cn}. Hence, it is open. 
 By Lemma~\ref{lemma:ine}, $y$ is  everywhere  injective on
 $\O$. By the Invariance of Domain Theorem $y$ is a homeomorphism 
 between $\O$ and $\O^y $. 
\end{proof}

\section{ \v Silhav\'y measure and perimeter: Proof of Theorem \ref{reverse}}\label{sil2}
%
 Within this section, $\O$ is assumed to be  an open subset of
$\R^n$, $n\ge 2$.  In particular, we are not restricting here to
$n=3$.  We are interested in properties of Sobolev homeomorphisms
$y$
in relation to sets of finite perimeter.
 In case $y$ is bi-Lipschitz, sets of finite perimeter are mapped onto
 sets of finite perimeter, see \cite[Theorem 3.16]{AFP}  whereas
 the same property does not hold for  $y$  in $W^{1,p}$ with
 $p<\infty$.  The aim of this
 section is that of proving  Theorem \ref{reverse}, which characterizes pairs
  $(y,E)$ ($y$ is a Sobolev map and $E\subset\Omega$ is a
  measurable set) such that $E^y$  is of finite perimeter in
  $\Omega^y$. We start by preparing some preliminary result.




\begin{proposition}[Perimeter = total variation of the \v{S}ilhav\'y measure]\label{direct}
 Assume  that $E\subset\Omega$ is measurable,   $y\in
W^{1,n}_{\rm loc}(\Omega;\R^n)$ is a homeomorphism, 
and there exists a vector Radon measure $\bp_{y,E}\in\mathcal{M}_{\rm loc}(\Omega;\R^n)$  such that \eqref{eq:p_yE-1} holds.
Then, $\per(E^y,\Omega^y)=|\bp_{y,E}|(\Omega)$.  In particular, if we
assume  that  $\bp_{y,E}$  is  finite, we get that the
perimeter of $E^y$ in $\Omega^y$ is finite  as well. 
\end{proposition}
\begin{proof}

A homeomorphism  in $W^{1,n}_{\rm loc}(\Omega;\R^n)$
 satisfies the Lusin's condition $N$ \cite[Thm. 3]{reshetnyak} 
 and is  almost-everywhere  differentiable \cite[Cor. 2.2.5]{HK}.
Thanks to the Lusin's condition $N$, the area formula holds with equality and gives
\[\begin{aligned}
\per(E^y,\O^y)&=\sup\left\{\int_{E^y}\textrm{div}\varphi(\xi)\, \d \xi \;|\;\varphi\in C^\infty_{\rm c}(\O^y;\R^n),\;\|\varphi\|_\infty\leq 
1\right\}\\&
=\sup\left\{\int_{E}\textrm{div}\varphi(y(x))\,\textrm{det}\nabla y(x)\,\d 
x\;|\;\varphi\in C^\infty_{\rm c}(\O^y;\R^n),\;\|\varphi\|_\infty\leq 1\right\}.
\end{aligned}
\]
 Note that the 
identity \begin{equation}\label{cofidentity}(\textrm{div}\varphi)\circ
  y\det\nabla y=\cof\nabla y:\nabla(\varphi\circ y)\end{equation}
holds almost everywhere in $\Omega$. 
Indeed, we  may  write $\mathrm{div} \varphi=\nabla\varphi:I$
(where $I$ is the identity matrix), and  relation  \eqref{cofidentity} 
follows from  the chain-rule  formula $\nabla(\varphi\circ
y)=(\nabla\varphi\:\circ y) \nabla y$,  which is valid almost
everywhere in  $\Omega$, and  from  the matrix identity $(\cof {A}) {A}^T=I\det {A}$.
Therefore, we get
\begin{equation}\label{eq:chain}\begin{aligned}
\per(E^y,\Omega^y)
=\sup\left\{\int_{E}\cof(\nabla y):\nabla(\varphi\circ y)\,\d x\;|\;\varphi\in 
C^\infty_{\rm c}(\O^y;\R^3),\;\|\varphi\|_\infty\leq 1\right\}.
\end{aligned}
\end{equation}
 As  $y\in W^{1,n}_{\rm loc}(\Omega;\R^n)$,  we have $\cof\nabla y\in
L^{r}_{\rm loc}(\Omega)$  with  $r=n/(n-1)$. 
Formula \eqref{eq:p_yE-1} can be extended by continuity to 
all test functions  in the class $W^{1,n}(\O;\R^n)\cap C_{\rm c}^0(\O;\R^n)$
since $\bp_{y,E} $ is a measure and the conjugated exponent of $r$ is $n$.
 Fix now $ \varphi\in 
C^\infty_{\rm c}(\O^y;\R^3)$ and notice that  there holds $\varphi\circ y \in 
C^0_{\rm c}(\O;\R^n)$,  as  $y$ is a homeomorphism   and hence
  $y^{-1}(\mathrm{supp}(\varphi))$ is compact  in   $\Omega$.  
 Moreover, since $y\in W^{1,n}_{\rm loc}(\Omega;\R^n)$,   we have
 that   $\varphi\circ y \in W^{1,n}(\O;\R^n)$. 
 Therefore, $\varphi\circ y$ is an admissible test function for 
 equality  \eqref{eq:p_yE-1}.

 From  \eqref{eq:chain} and the extension of \eqref{eq:p_yE-1} to $W^{1,n}(\O;\R^n)\cap C_{\rm c}^0(\O;\R^n)$ we obtain
\begin{equation}\label{plus}
	\per(E^y,\O^y)=\sup\left\{\int_{\O}(\varphi\circ y)\cdot \d 
\bp_{y,E}\;|\;\varphi\in C^\infty_{\rm c}(\O^y;\R^n),\;\|\varphi\|_\infty\leq 1\right\}.
\end{equation}
On the other hand,  the total variation of ${p}_{y,E}$ is, by definition,
\begin{equation}\label{eq:tot-var}\begin{aligned}
 |\bp_{y,E}|(\O)&=\sup\left\{\int_{\O}f\cdot \d 
\bp_{y,E}(x)\;|\;f\in C^0_{\rm c}(\O;\R^n),\;\|f\|_\infty\leq 1\right\}.
\end{aligned}\end{equation}
From \eqref{plus} and \eqref{eq:tot-var}  it immediately follows that
\begin{equation}\label{eq:per-ineq}
\per(E^y,\O^y)	\leq |\bp_{y,E}|(\O).
\end{equation}

 In order to establish the   reverse inequality, one has to
prove that any $f\in C^\infty_{\rm c}(\O;\R^n)$ can be  uniformly  
approximated  by functions of the form $\varphi\circ y $, with $ \varphi\in 
C^\infty_{\rm c}(\O^y;\R^n)$.  Fix  $f\in C^0_{\rm c}(\O;\R^n)$ and $K:=\mathrm{supp}(f)$. Then 
$K^y$ is compact in $\O^y$. On $K^y$, define the function
$
g:=f\circ y^{-1}, 
$ which can be extended to $g\in C^0_{\rm c}(\O^y;\R^n)$ by setting $g=0$
outside $K^y$.  For all $\e>0$ choose now $\varphi_\e\in
C^\infty_{\rm c}(\O^y;\R^n)$ with $\sup_{\Omega^y}|g-\varphi_\e|<\e$. Then,
one has that  
$
\sup_\Omega|f-\varphi_\e\circ y|<\e,
$
which provides the desired approximation. 
\end{proof}

 The proof of Theorem \ref{th2} follows from checking the converse
statement of Proposition \ref{direct}. In order to achieve this,  a crucial role is played by the following result on Sobolev homeomorphisms of finite distorsion due to Cs\"ornyei, Hencl, and Mal\'y \cite{CHM}, see also \cite{HK0, HKM}. 
\begin{proposition}[$\mbox{\cite[Theorem 1.2]{CHM}}$]
\label{simple}
Let $y\in W^{1,n-1}_{\rm loc}(\Omega; \R^n )$ be  a
homeomorphism of finite distorsion. Then $y^{-1}\in W^{1,1}_{\rm
  loc}(\Omega^y;\R^{ n})$ and is of finite distorsion.
\end{proposition}
Taking advantage of the latter result, we  can now proceed to the  proof of Theorem \ref{reverse}. 
\begin{proofreverse}
 Given Proposition \ref{direct}, we are left with the converse
statement. Namely, for all $E\subset \O$ measurable and all $y \in
W^{1,n}_{\rm loc}(\O;\R^n)$ homeomorphism of finite distorsion with
$\per(E^y,\Omega^y)<\infty$ we should find a finite Radon measure (the
\v{S}ilhav\'y measure) such that relation \eqref{eq:p_yE-1} holds. 

Let $\psi\in C^{\infty}_{\rm c}(\Omega;\R^n)$  with 
$\|\psi\|_\infty\le 1$  be given.  
Since $y$ is a homeomorphism, we have  that  $\psi\circ y^{-1}\in C^0_{\rm c}(\Omega^y;\R^n)$.
By  Proposition   \ref{simple}, we also get $\psi\circ y^{-1}\in W^{1,1}(\Omega^y;\R^n)$.
 Let $\eps>0$ and $\varphi_\eps\in C^\infty_{\rm c}(\Omega^y;\R^n)$ be
 defined by $\varphi_\eps:=(\psi\circ y^{-1})\ast\rho_\eps$, where
 $\rho_\eps(x)=\eps^{-d}\rho(x/\eps)$ and $\rho$ is the standard unit
 symmetric mollifier in $\R^n$. Notice that,  by choosing  $\e_0$
 small enough one has   that  the support of $\varphi_\eps$ is
 compact in $\Omega^y$ for any $0<\eps<\eps_0$. Moreover,
 $\|\varphi_\eps\|_\infty \le 1$ and $\varphi_\eps$ converge strongly
 to $\psi\circ y^{-1}$ in $W^{1,1}(\Omega^y;\R^n)$ as $\eps\to
 0$.  As $y$
 satisfies the Lusin's condition $N$ the area formula holds with
 equality,  hence 
 \begin{equation}\label{luz}
 \int_{E^y}\mathrm{div}(\psi\circ y^{-1})\ \d \xi=\int_{E^y}I:
 (\nabla\psi)\circ y^{-1} \,\nabla y^{-1} \, \d \xi=\int_E (\det\nabla
 y) \,I:\nabla\psi\,(\nabla y^{-1}\circ y) \, \d x.
 \end{equation}
Since $ \nabla y^{-1}(y(x))=(\nabla y(x))^{-1}$ holds at any
differentiability point $x$ of $y$ such that $\det\nabla y(x)>0$,
hence  almost everywhere  in the set $\{\det\nabla y>0\}$, from \eqref{luz} we deduce
\begin{equation}\label{I:}\begin{aligned}
\int_{E^y}\mathrm{div}(\psi\circ y^{-1})\, \d \xi&=\int_{\{\det\nabla
  y>0\}}(\det\nabla y)\,I:\nabla \psi \,(\nabla y)^{-1}\, \d x
\\&=\int_{\{\det\nabla y>0\}}\det\nabla y \,(\nabla y)^{-T}:\nabla
\psi \, \d x=\int_{E}\cof\nabla y:\nabla \psi\, \d x.
\end{aligned}\end{equation}
Notice that the last equality in \eqref{I:} follows from  the fact
that $y$ is of finite distorsion,   which implies $\cof\nabla y=0$
almost everywhere on $\{\det\nabla y=0\}$.
	Similarly, by the area formula and by \eqref{cofidentity}  we obtain
	\begin{equation}\label{similar}\begin{aligned}
	\int_E \cof\nabla y: \nabla (\varphi_\eps\circ y)\, \d x&=\int_E \det\nabla y\: I:(\nabla\varphi_\eps)\circ y \, \d x=\int_{E^y}\mathrm{div}\varphi_\eps \, \d \xi.
	\end{aligned}\end{equation}
	Since $\mathrm{div}\varphi_\eps$  converges  to $\mathrm{div}(\psi\circ y^{-1})$ in $L^1(\Omega^y)$ as $\eps\to 0$, from  \eqref{similar} we get
	\[
	\lim_{\eps\to 0}\int_E \cof\nabla y: \nabla (\varphi_\eps\circ
        y)\, \d x=\int_{E^y}\mathrm{div} (\psi\circ y^{-1})\, \d \xi.
	\]
	By combining the latter with \eqref{eq:chain} and \eqref{I:}, with we deduce
	\[\begin{aligned}
	\int_{E}\cof\nabla y:\nabla\psi\, \d x&=\lim_{\eps\to 0}\int_E
        \cof\nabla y: \nabla (\varphi_\eps\circ y)\, \d x\\
	&\le \sup\left\{\int_E \cof\nabla y: \nabla (\varphi\circ y)\,
          \d x\ | \  \varphi\in C^\infty_{\rm c}(\Omega^y;\R^n),\;\|\varphi\|_\infty\le 1\right\}\\&=\per(E^y,\Omega^y).
	\end{aligned}\]
	  We have hence checked that 
	 \[
	 \sup\left\{\int_E \cof\nabla y: \nabla \psi\, \d x \ |\  \psi\in C^\infty_{\rm c}(\Omega;\R^n),\;\|\psi\|_\infty\le 1\right\}\le\per(E^y,\Omega^y)<\infty.
	 \]
	  This implies  that the distributional divergence of $\chi_E \cof\nabla y$ is a finite measure on $\Omega$.
	\end{proofreverse}

%

The \v{S}ilhav\'y measure $\bp_{E,y}$ given by Theorem \ref{reverse} is the
distributional divergence of $-\chi_E\cof\nabla y$. Therefore,  in
order to  have that   $\per(E^y,\O^y)<\infty$, Theorem
\ref{reverse} requires  $\chi_E\cof\nabla y$ to be a divergence
measure field.  By strengthening the assumptions one may obtain
 improved characterizations of the divergence of such  fields, see
for instance \cite{ACM, ctz, S1}.  In particular, we can prove
 the following.

\begin{proposition}[Support of the \v{S}ilhav\'y measure] 
  Under the  assumptions of  Proposition \emph{\ref{direct}} 
  let  $\per(E,\O)<\infty$. Then, $\bp_{E,y}$ is concentrated on the closure of the reduced boundary of $E$ in ~$\Omega$.
\end{proposition}
\begin{proof}
Let $y_\eps:=y\ast\rho_\eps$, with
$\rho_\eps(x)=\eps^{-d}\rho(x/\eps)$ and $\rho$  be  the standard
mollifier.  Since  $y_\eps$ is smooth,  $\chi_E$ is a function
of bounded variation, and the cofactor is divergence-free,  we
readily have that 
$\mathrm{div}(\chi_E\cof\nabla y_\eps)=\cof\nabla y_\eps\nabla \chi_E$ is a measure concentrated on the reduced boundary of $E$ in $\Omega$.
Notice that $\cof\nabla y_\eps$  converges to  $\cof\nabla y$
in $L^{n/(n-1)}_{\rm loc}(\Omega; \R^{n\times n})$, so that integration by parts entails
\[
-\int_\Omega\psi\cdot \md(\mathrm{div}(\chi_E\cof\nabla y_\eps))
=\int_{\Omega}\chi_E\cof\nabla y_\eps:\nabla\psi\, \d x
\to\int_E\cof\nabla y:\nabla\psi\, \d x=\int_\Omega\psi\cdot \d\bp_{E,y}
\]
as $\eps\to 0$, for every $\psi\in
C^\infty_{\rm c}(\Omega;\mathbb{R}^n)$.  For all $\eps>0$, the measure
 $\mathrm{div}(\chi_E\cof\nabla y_\eps)$ is concentrated on the
reduced boundary of $E$ in $\Omega$. We hence conclude that   $\bp_{E,y}$ is concentrated on the closure of the reduced boundary.
\end{proof}

 In case $y^{-1}\in W^{1,n}_{\rm loc}(\Omega^y;\R^n)$ the
 characterization of Theorem \ref{th2} can be applied to the inverse
 deformation $y^{-1}$. Note that such regularity of the inverse follows for
 instance for mappings with  $L^{n-1}$ distorsion, see \cite{HKM}. Therefore, We have the
 following 
	
	\begin{corollary}[Characterization for the inverse deformation] Suppose that $E\subset\Omega$ is a
          measurable set and that $y\in W^{1,n}_{\rm
            loc}(\Omega;\R^n)$ is a homeomorphism of finite distorsion
           with $K_y\in L^{n-1}(\Omega)$. 
	Then, $\per(E,\Omega)<\infty$ if and only if the distribution $\bp_{E^y,y^{-1}}:=-\mathrm{div} (\chi_{E^y}\cof\nabla y^{-1})$ is a finite Radon measure on $\Omega^y$.
	\end{corollary}
%
%
\section{ Existence of minimizers: Proof of Theorem \ref{th1}}\label{conve}
%
 The aim of this section is to discuss the existence of minimizers
of both $\mathcal{F}_0$ and $\mathcal{F}_\eps$ on the respective sets
of admissible deformations. This in
particular proves Theorem \ref{th1} as well as the existence statement
in Theorem \ref{th2}.

We
start by establishing some preliminary result  on the convergence
of the deformed domains and phase configurations associated to a
$\bbY$-converging sequence of deformations. A  crucial tool in
this direction is the  semicontinuity of the perimeter in the
deformed configuration, when both the ambient sets $\O^{y_k}$ and the
finite perimeter sets $F_k\subset \O^{y_k}$ vary along a sequence, see
 Proposition  \ref{lem:2}. This will prove to be essential for the
$\Gamma$-limit result stated in Section \ref{lastsection}.  

We shall make use of the following equiintegrability result for
inverse Jacobians of mappings of integrable distorsion, which  is
inspired by the work of  Onninen and Tengvall \cite{OT}.
\begin{lemma}[Equiintegrability of $\det\nabla y_k^{-1}$]\label{lem:equi}
	Let $y_k:\Omega\to\Omega^{y_k}$ be  homeomorphisms with  uniformly
        $L^q$-integrable distorsion  for $q > 2$ (namely,
        $\|K_{y_k}\|_{L^q(\O)}$ is bounded independently of $k$). Then, $\det\nabla y_k^{-1}$ are equiintegrable on $\Omega^{y_k}$.
\end{lemma}
\begin{proof} 
	 From \cite[Theorem 1.4]{OT} we have  that
	\[
	\int_{\Omega^{y_k}} |\nabla y_k^{-1}|^3\log^s({\rm e}+|\nabla y_k^{-1}|)\,\d\xi\le C\int_\Omega K_{y_k}^q\,\d\xi,
	\]
	where $s=2(q-2)$ and $C$ is a constant depending only on $q$. Notice that by the elementary inequality $|\det  F|\leq 6| F|^3 $,  we have
	\[
	|\nabla y_k^{-1}|^3\log^s({\rm e}+|\nabla y_k^{-1}|)\ge \frac16 3^{-s} \det\nabla y^{-1}_k\log^s\left({\rm e}^3+\frac16\det\nabla y_k^{-1}\right).
	\]
	We conclude that 
	\[
	\int_{\Omega^{y_k}}\det\nabla y_k^{-1}\log^s\left({\rm e}^3+\frac16\det\nabla y_k^{-1}\right)\,\d\xi\le  C' \int_\Omega K_{y_k}^q\,\d\xi,
	\]
	where $ C'$ depends only on $q$.  The latter right-hand
        side is uniformly bounded. This entails that the
        superlinear function of the determinant on the left-hand side
        is uniformly bounded as well. This implies the
        equiintegrability of the sequence of the determinants of the inverses. 
\end{proof}

\begin{lemma}[Convergence of deformed configurations]\label{lem:0}
	Let $y,y_k\in\bbY$ such that $y_k\to y$ weakly in $W^{1,p}$,
        $p>3$ (hence uniformly).  Then, 
	\begin{itemize}
		\item [(i)] For any open sets $A,O$ such that
                  $A\subset\subset\O^y\subset\subset O$, one has
                  $A\subset\O^{y_k}\subset O$  for  $k$ large
                  enough. In  particular, $|\O^y\Delta\O^{y_k}|\to 0 $;
		\item [(ii)] If $\|K_{y_k}\|_{L^q(\O)}\leq c $ 
                  uniformly,  by  letting $O^k:=\O^y\cap\O^{y_k}$, there holds
		$$
		|\O\setminus(y^{-1}(O_k)\cap y_k^{-1}(O_k))|\to 0.
		$$
	\end{itemize}
\end{lemma}
\begin{proof}
{ Ad	(i):}   Let $V$ be open and  such that
$A\subset\subset V\subset\subset\O^y $. Since $\overline{A}$ and
$\partial V$ are disjoint compact sets,  we have that 
$d(\overline{A},\partial V)=:2\delta>0$.  Let
$U=y^{-1}(V)\subset\subset\O$ and  $V_k=y_k(U)$. Since $y,y^k\in\bbY$
are homeomorphisms on $\overline{U}$, we have $\partial V=y(\partial
U)$ and $\partial V_k=y_k(\partial U)$.  As $p>3$,   we have that
 $y_k\to y$ in $C(\overline{\O};\Rz^3)$,   thus 
$\|y-y_k\|_\infty<\delta$  for $k$ large enough.   Hence, for
any boundary point $\xi\in\partial V_k$,   we have that 
$d(\xi,\partial V)<\delta$  for $k$ large,  which yields
$\overline{A}\subset V_k\subset\O^{y_k} $ owing to
$d(\overline{A},\partial V)=2\delta $.

 As  $O\supset\supset \O^y$ we deduce as above  that 
$d(\partial O,\overline{\O^y})=:2\delta$ for some $\delta>0$, which
immediately yields the inclusion $\overline{\O^y}+B(0,\delta)\subset
O$. Then,  since $\|y-y_k\|_\infty<\delta$ we have that
 $$ \O^{y_k}  \subset\O^y+B(0,\delta)  \subset \overline{\O^y} +
B(0,\delta) \subset O.$$

 In order to check that  $|\O^y\Delta\O^{y_k}|\to 0 $, observe
that $\O^y$ can be approximated in measure by open sets $ A_\ell
\subset\subset\O^y$ ($\O^y$ can be approximated by  internal
 compact sets). Moreover, $\overline{\O^y}$ can be approximated in
measure by external open sets $ O_\ell \supset\overline{\O^y}
$.  Since 
	$\O$ is a bounded Lipschitz domain, by Lusin's $N$ property
         for  (a $W^{1,p}$  extension of) $y$ and the fact $y(\partial\O)\supset\partial(\O^y) $, it follows that $|\partial(\O^y)|=0$, i.e. $|\overline{\O^y}|=|\O^y| $.

{ Ad (ii):}  Since $y^{-1}(O_k)\subset \O $ and  $ y_k^{-1}(O_k)\subset \O$, it is sufficient to prove
	$$
	|y^{-1}(O_k)|\to |\O|,\qquad |y_k^{-1}(O_k)|\to |\O|.
	$$
	Firstly,   $|\O^y\setminus O^k|\to 0$ by (i). Hence, since $\det\nabla y^{-1}\in L^1(\Omega^y)$,
	$$
	|y^{-1}(O_k)|=\int_{O^k}\det\nabla y^{-1}\,\d\xi\to \int_{\O^y}\det\nabla y^{-1}\,\d\xi=|\O|.
	$$
	Secondly,
	\begin{align*}
	|y_k^{-1}(O_k)|=\int_{O^k}\det\nabla y_k^{-1}\,\d\xi=\int_{\O^{y_k}}\det \nabla y_k^{-1}\,\d\xi-\int_{\O^{y_k}\setminus O^k}\det\nabla y_k^{-1}\,\d\xi=\\
	=|\O|-\int_{\O^{y_k}\setminus \O^y}\det\nabla y_k^{-1}\,\d\xi.
	\end{align*}
	 By  Lemma \ref{lem:equi}, the determinants $\nabla
        y_k^{-1} $ are  equiintegrable. 
	Since  $ |\O^{y_k}\setminus \O^y|\to 0$, the statement follows.
\end{proof}
\begin{lemma}[Convergence of the phases]\label{lem:1}
	Let $y,y_k\in\bbY$ such that $y_k\to y$ weakly in $W^{1,p}$,
        for $p>3$, and have uniformly $L^q$-bounded distorsion, for $q >
        2$.   Let $\z\in L^\infty(\O^{y},[0,1])$ and $\z_k\in L^\infty(\O^{y_k},[0,1]) $. Finally, let $z=\z\circ y,z_k=\z_k\circ y_k\in L^\infty(\O;[0,1]) $  and $O_k:=\O^y\cap\O^{y_k}$. Then,
	$$
	\|\z-\z_k\|_{L^1(O_k)}\to 0\;\;\Rightarrow \;\;	\|z-z_k\|_{L^1(\O)}\to 0.
	$$
\end{lemma}
\begin{proof}  By  introducing the shorthand
  $E_k:=y^{-1}(O^k)\cap y_k^{-1}(O^k)\subset\O$,  we start by
  observing that 
	$$
	\|z_k-z\|_{L^1(\O)}\leq |\O\setminus E_k|+\|z_k-z\|_{L^1(E_k)}.
	$$
	 As   $|\O\setminus E_k|\to 0$ by Lemma
        \ref{lem:0},  we are left to prove that
        $+\|z_k-z\|_{L^1(E_k)}\to 0$.  One uses  the triangle inequality to write
	$$
	\|z_k-z\|_{L^1(E_k)}\leq I_k^{(1)}+I_k^{(2)},
	$$
	with
	\begin{align*}
	I_k^{(1)}:=\|\z_k\circ y_k-\z\circ y_k\|_{L^1(E_k)},\qquad I_k^{(2)}:=\|\z\circ y_k-\z\circ y\|_{L^1(E_k)}
	\end{align*}
 The  $L^q$-bound on the  distortion and Lemma \ref{lem:equi}
 entail that   the sequence  $\det\nabla y_k^{-1} $ is 
equiintegrable.  Let  $\rho:[0,+\infty) \to [0,+\infty)$
(monotonically increasing) be a modulus of  equiintegrability for
the sets  $\{\det\nabla y_k^{-1}\}_{k\ge 1}\cup \{\det\nabla
y^{-1}\} $, i.e., for any measurable set  $A\subset\Rz^3$  we ask
for $\lim_{t\to 0+}\rho(t)=0 $ and 
	\begin{align*}
	\int_{ \O^{y}\cap A}\det\nabla y^{-1}\,\d
        \xi\,\vee\,\sup_{k}\int_{ \O^{y_k}\cap A}\det\nabla
        y_k^{-1}\,\d \xi\le \rho(|A|) .
	\end{align*}
	Now, fix $\delta>0$  and  change  variable
        $x\mapsto\xi$ in the integral in $ I_k^{(1)} $  getting 
	\begin{align}\label{eq:I1}
	I_k^{(1)}=\int_{ y_k(E_k)}\det\nabla y_k^{-1}|\z_k-\z|\,\d \xi\leq \rho(|A_k(\delta)|)+\delta|\O|,  
	\end{align}
	where  $ A_k(\delta):=\{\xi\in O^k\ | \ |\z_k(\xi)-\z(\xi)|>\delta\}$. 
 Since  $\|\z-\z_k\|_{L^1(O_k)}\to 0$  one has that 
$|A_k(\delta)|<\delta$  for $k$ large enough. 

 In order to control  $I_k^{(2)}$, let $\overline{\z}_\delta\in C^0(\overline{\O^y},\Rz)$ be a (uniformly) continuous  $L^1$ approximation of $\z$ such that $
	\|\overline{\z}_\delta-\z\|_{L^{1}(\O^y)}
	$ is so small that 
	$$
	|\overline A(\delta)|<\delta \ \  \text{for} \ \  \overline
        A(\delta):=\{\xi\in \O^y\ | \ |\overline{\z}_\delta(\xi)-\z(\xi)|>\delta\}  
	$$	
	We write  $	I_k^{(2)}\leq 	J_k^{(1)}+	J_k^{(2)}+	J_k^{(3)}$, with
	\begin{align*}
	J_k^{(1)}=\|\z\circ y_k-\overline{\z}_\delta\circ y_k\|_{L^1(E_k)},\;
	J_k^{(2)}=\|\overline{\z}_\delta\circ y_k-\overline{\z}_\delta\circ y\|_{L^1(E_k)},\;
	J_k^{(3)}=\|\overline{\z}_\delta\circ y-\z\circ y\|_{L^1(E_k)}.
	\end{align*}
	Now, similarly to \eqref{eq:I1}, we can write
	\begin{align}\label{eq:J1}
	J_k^{(1)}+J_k^{(3)}\le 2\rho(|\overline A(\delta)|)+2\delta|\O|\leq 2\big(\rho(\delta)+\delta |\O|\big).
	\end{align}
	Finally, since $\overline{\z}_\delta$ is uniformly continuous and $|\O|<+\infty$, if $\omega_\delta$ is the modulus of uniform continuity of $\overline{\zeta}_\delta$,  we get
	\begin{align}\label{eq:J2}
	J_k^{(2)}\leq \omega_\delta(\|y-y_k\|_\infty)|\O|.
	\end{align}
	By combining \eqref{eq:I1}  and  \eqref{eq:J2}  and
        using the fact that $\delta$ is arbitrary,  we  obtain the statement.
\end{proof}	
The following result concerns the semicontinuity of the perimeter of
sets in the deformed configuration along sequences of suitably
converging sets and deformations.  This is based on  the
characterization result from Theorem \ref{reverse}.
\begin{proposition}[Lower semicontinuity of the perimeter]\label{lem:2}
	Let $(y_k,\z_k)\in\bbQ,\,y\in\bbY,\,\z\in
        L^\infty(\O^{y},\{0;1\})$ with $y,y_k$ satisfying the
        assumptions of  Lemma \emph{\ref{lem:1}}.   Let $F=\{
        \xi\in\O^y\ | \ \z(\xi)=1\}$, $F=\{ \xi\in\O^{y_k}\ | \ \z_k(\xi)=1\}$ and assume $ |F_k\Delta F|\to 0$. If $I:=\liminf_{k\to+\infty}\per(F_k,\O^{y_k})<\infty$, then
	$$
	\per(F,\O^y)\leq I\quad {\rm and}\quad  (y,\z)\in\bbQ ;
	$$ 
\end{proposition}
\begin{proof}
	Letting $ E=y^{-1}(F)$, and $E_k=y_k^{-1}(F_k)$, we have by
        Theorem \ref{reverse}  that 
	$$
	\per(F_k,\O^{y_k})=|\bp_{y_k,E_k}|.
	$$	
	By  applying  Lemma \ref{lem:1}  to  $\z=\chi_F,
        \z_k=\chi_{F_k}$  we deduce that  $\chi_{E_k}\to
        \chi_E$ in  $L^1(\O)$.  Moreover, since $\nabla y_k\to
        \nabla y$ weakly in  $L^p(\O)$, the convergence 
 	$
	\cof\nabla y_k\to\cof \nabla y
	$  holds  weakly in  ${L^{p/2}}(\O)$. 
	Therefore, for any test function $\psi\in C^\infty_{\rm c}(\O;\R^3)
        $, as $k\to\infty$ we have 
	$$
	\int_{\O}\psi\cdot \d \bp_{y_k,E_k}=\int_{\O}\chi_{E_k}\cof\nabla y_k:\nabla\psi\,\d x\to \int_{\O}\chi_{E}\cof\nabla y:\nabla\psi\,\d x=: \bp_{y,E}(\psi),
	$$
	where the last equality is a definition of the distribution on
        the right side. By the lower semicontinuity of the total
        variation,  we have that $|\bp_{y,E}|\leq I$. We conclude
        by 
        Theorem~\ref{reverse} as $\per(F,\O^y) = |\bp_{y,E}|$. 
\end{proof}

 After this preparatory discussion, we eventually move to the
existence proof for minimizers.   
First we show that the diffuse-interface functional  $\mathcal
F_\eps$  admits a minimizer for every $\varepsilon>0$.  Such
existence  result
is part of the statement of Theorem \ref{th2}.  Indeed, we  restate it here in a
slightly more general form, in which the Dirichlet boundary condition
is imposed only on a subset of the boundary of positive
$\HH^2$-measure, as  it is customary  in elasticity theory.

\begin{proposition}[Existence for the diffuse-interface model]\label{Prop:diffuse-ex}
 Under assumptions \eqref{W-ass1}-\eqref{W-ass4}, let
$\Gamma_0\subset\partial\O$ be  relatively open in $\partial \O$
with 
$\HH^2(\Gamma_0)>0$. Moreover, let $\epsi>0$ and  $(y_0,\zeta_0)\in
\mathbb{Y}\times W^{1,2}(\O^y;[0,1])$ be
such that the  set  
        ${\widetilde{\mathbb{Q}}}_{(y_0,\Gamma_0)}:= \{(y,\z) \in
        \mathbb{Y} \times   W^{1,2}(\O^y;[0,1])\ | \
          y=y_0\text{ on $\Gamma_0$}\}$ is nonempty and
          $\FF_\eps(y_0,\zeta_0)<\infty$.  Then,  there is a minimizer of $\FF_\eps$ on ${\widetilde{\mathbb{Q}}}_{(y_0,\Gamma_0)}$.
\end{proposition}

\begin{proof}
	Let $(y_k,\z_k)\in {\widetilde{\mathbb{Q}}}_{(y_0,\Gamma_0)}$
        be a minimizing sequence  for  $\FF_\eps$. 
	 The coercivity \eqref{W-ass2}  and the 
        generalized Friedrichs inequality imply that  one can extract
        a  not relabeled subsequence  such
that $y_k\to y$ weakly  in $W^{1,p}(\O;\R^3)$.  The boundary
condition and the 
Ciarlet-Ne\v{c}as condition \eqref{inj} are readily preserved in the
limit. Moreover, one has that the distorsion $K_y\in L^q(\Omega)$ as
the function $F \to |F|^3/\det F$ is polyconvex and $F_k = \nabla y_k$
are weakly converging. We conclude that $y\in \mathbb{Y}$ and $y=y_0$
on $\Gamma_0$. 

	For every $\delta>0$, let $O_\delta:=\{\xi\in\O^y|\, {\rm
          dist}(\xi,\partial\O^y)>\delta\}\subset\subset\O^y$. By Lemma
        \ref{lem:0}  we have that  $\O^y=\cup_\delta O_\delta$
        and $ O_\delta\subset \O^{y_k}$
         for  $k$  large. 
	Denote by $\eta_k $ and $H_k$ the  trivial  extensions on $\Rz^3$  of  $\z_k$ and $\nabla \z_k$ respectively.
	 The coercivity  of $\FF^{\rm int}_\eps$ implies that
         one can extract  not relabeled subsequences such that 
	$\eta_k\to \eta$ weakly* in $L^\infty(\R^3)$   and $H_k\to H$
        weakly in $L^2(\R^{3})$.  Set now  $\z:=\eta|_{\O^y}$. 
	For every $\xi_0\in O_\delta$ and $B(\xi_0,r)\subset O_\delta$ we
        have  that 
	$\eta_k\to \eta$ weakly in $W^{1,2}(B(\xi_0,r))$. This implies
        that $H=\nabla\eta=\nabla\z $ almost everywhere in
        $B(\xi_0,r)$. Moreover,  by possibly extracting again one
        has that   $\eta_k\to\eta $ strongly in
        $L^2(B(\xi_0,r))$. 
	 As every $\xi\in\O^y$ belongs to some $O_\delta$ for $\delta$
         small enough, we get that   $H=\nabla \z$ almost
         everywhere~in $\O^y$. It is also easy to see that $\eta=0,$
         $H=0$  almost everywhere~on the complement of $\O^y$ due to
         the uniform convergence of $y_k$.  Indeed, if 
         $\xi_0\not\in \overline{\O^y}$ then there are two open disjoint
         neighborhoods of $\xi_0$ and $\overline{\O^y}$. Let
         $O\supset\overline{\O^y}$ be the open neighborhood of
         $\overline{\O^y}$. Then for $k$ large enough $\O^{y_k}\subset
         O$ (Lemma \ref{lem:0}), i.e. $\eta_k=0$, $ H_k=0$ in a
         neighborhood of $\xi_0$. Consequently, $ \eta(\xi_0)=0$,
         $H(\xi_0)=0$ at least if $\xi_0$ is a Lebesgue point of $\eta$
         and $H$.  

 The latter argument shows that $\eta_k\to\eta $  pointwise almost
everywhere in the complement of $\overline{\O^y}$. Up to possibly extracting again, we hence have that 
$\eta_k\to\eta $ pointwise almost everywhere in $\Rz^3$ as well. In
fact, the pointwise convergence in $\overline{\O^y}$ follows since
$\eta_k\to\eta $ strongly in $L^2(B(\xi_0,r))$ for any
$B(\xi_0,r)\subset\subset\O^y$ and  $|\eta-\eta_k|\leq 1$ almost
everywhere. 

	 Using the Fatou Lemma, we  find 
	\begin{align}
	\liminf_{k\to\infty}\mathcal{F}^{\,\rm int}_\eps(y_k,\z_k)&=\liminf_{k\to\infty}\int_{\R^3} \Big( \frac\eps2|H_k|^2+\frac1\eps \Phi(\eta_k)\Big)\,\d\xi\ge \int_{\R^3} \Big( \frac\eps2|H|^2+\frac1\eps \Phi(\eta)\Big)\,\d\xi\nonumber\\
	&=\int_{\O^y} \Big( \frac\eps2|\nabla \z|^2+\frac1\eps \Phi(\z)\Big)\,\d\xi=\mathcal{F}^{\,\rm int}_\eps(y,\z)\label{prove}
	\end{align}
	 which  shows the weak lower semicontinuity of the
        interfacial energy.  

	To show the weak lower semicontinuity of the bulk
        contribution, we write  it  as 
	\begin{equation*}
	\widetilde{\mathcal{F}}^{\,\rm bulk}(y,z)=\int_\Omega\Big( z(x)
        W_1(\nabla y(x))+(1-z(x))W_{ 0}(\nabla y(x))\Big)\,\d x  ,
	\end{equation*}
	Notice that the integrand is continuous in $z$ and convex in
        $\nabla y$ and  in its minors.  Let now
        $z_k:=\z_k\circ y_k$ and recall from 
        Lemma~\ref{lem:1}  entails that  $z_k\to z=\z\circ y$ in $L^1(\O)$.
By applying  \cite[Cor.~7.9]{fonseca-leoni}
	we get that $\liminf_{k\to\infty}
        \widetilde{\mathcal{F}}^{\,\rm bulk}(y_k,z_k)\ge
        \widetilde{\mathcal{F}}^{\,\rm bulk}(y,z)$. 
	Consequently, 
	\begin{equation}\label{semifel}
	\liminf_{k\to\infty}  \mathcal{F}^{\,\rm bulk}(y_k,\z_k)
        =\liminf_{k\to\infty}  \widetilde{\mathcal{F}}^{\,\rm
          el}(y_k,z_k) \ge \widetilde{\mathcal{F}}^{\,\rm bulk}(y,z) =
        \mathcal{F}^{\,\rm bulk}(y,\z).
	\end{equation}
 Together with \eqref{prove}, the latter proves  that $(y,\z)$ is
a minimizer of $\FF_\eps$ on $ {\widetilde{\mathbb{Q}}}_{(y_0,\Gamma_0)}$ by means of the direct method \cite{dacorogna}.    
\end{proof}
 We conclude this Section by providing a proof of Theorem
\ref{th1}. 

\begin{proof}[Proof of Theorem \emph{\ref{th1}}]
%
 Let $(y_k,\z_k)\in\bbQ_{y_0}$ be a minimizing sequence  for
 $\FF_0$.  As in  the proof of Proposition
 \eqref{Prop:diffuse-ex}, we can assume, up to extraction of a  not
 relabeled  subsequence, that $y_k\to  y$ weakly in
 $W^{1,p} $  for some $y\in\mathbb Y$.

Letting $F_k=\{ \z_k=1\} $, we can identify the sequence of states
with $(y_k,F_k)$. Since  the interface energy is bounded along the
sequence $(y_k,F_k)$, the sets  $F_k$ have uniformly  bounded
perimeters,  namely,   $\per(F_k,\O^{y_k})\le c $. 
For   $\ell\in \Nz$, let $O^{\ell}:=\{x\in\O^y|\, {\rm
  dist}(x,\partial\O^y)>2^{-\ell}\}\subset\subset\O^y$.  As
$O^{\ell}\subset\O^{y_k}$ for $k$ large enough due to Lemma
\ref{lem:0}, for any given $\ell\in\Nz$  we have that $\limsup_k
\per(F_k,O^{\ell})\le c$.  We can hence   find a measurable
set $G^{\ell}\subset  O^{\ell}$ and a  not relabeled  subsequence $F_{h}$ such that
\[|(F_{h}\Delta G^{\ell})\cap O^{\ell}|\to 0\quad \textrm{for}\quad h\to\infty.
\]
 For all  $\ell'>\ell$  we can  further  extract 
a  subsequence $F_{h'}$ from  $F_{h}$ above  in
such a way  that $|(F_{h'}\Delta G^{\ell'})\cap O^{\ell'}|\to 0 $
and  $G^{\ell'}\cap O^{\ell}=G^{\ell} $.  From  the nested
family of subsequences corresponding to $\ell=1,2,\ldots$  we
extract by a diagonal argument a further subsequence $F_{k'}$. By
setting     $F:=\cup_\ell G^{\ell}$  and,  owing to $
O^{\ell}\nearrow\O^y$,  we get that  
\[
|(F_{k'}\Delta F)\cap \O^{y}|\to 0.
\]
Now, the set $F$ has finite perimeter in $\O^y$ as a consequence of
 Proposition  \ref{lem:2}.  By letting  $\z=\chi_F|_{\O^y} $  we
then have that  $(y,\z)\in \mathbb Q_{y_0} $. 
%
%

 One is left to check that $\FF_0 (y,\z) \leq \liminf
\FF_0(y_k,\z_k)$, which follows from the lower semicontinuity of
$\FF_0$. Indeed, the lower semicontinuity of bulk part of $\FF_0$
follows by the argument of  Proposition
\ref{Prop:diffuse-ex}.  As concerns the interface term, one just
needs to recall Proposition  \ref{lem:2}.
\end{proof}
%

%
\section{ Convergence of phase-field approximations: Proof of
  Theorem \ref{th2}}\label{lastsection}
%
 This section is devoted to the proof of the convergence Theorem
\ref{th2}. The argument relies on $\Gamma$-convergence
\cite{Bra02,DalMaso93}. In particular, we prove a $\Gamma$-$\liminf$
inequality for the interfacial part in Proposition \ref{theo:liminf}
and construct a recovery sequence in Proposition
\ref{theo:limsup}. Let us start by the former. 

\begin{proposition}[$\Gamma$-$\liminf$ inequality]\label{theo:liminf}
	Let $(y_k,\zeta_k), (y, {\zeta})\in\overline\bbQ$ be such that
	\begin{itemize}
		\item [i)] $ \liminf_{k\to+\infty}  \mathcal{F}^{\,\rm
                    int}_{\eps_k}(y_k,\z_k)<\infty$ for  some  sequence $\eps_k\to 0$,
		\item[ii)] $y_k\to y$ weakly in $W^{1,p}(\O;\Rz^3)$, $p>3$,
		\item [iii)]  $\lim_{k\to+\infty}\|\zeta_k- {\zeta}\|_{L^1(O^k)}=0 $, with $ O^k:=\O^{y_k}\cap \O^y$.
	\end{itemize}
	Then, there exists $E^y\subset \O^y$ measurable such that
	$$
	{\z }=\chi_{E^y} \ \  \text{and}  \ \ \gamma\per(E^y,\O^y)\leq \liminf_{k\to+\infty}  \mathcal{F}^{\,\rm int}_{\eps_k}(y_k,\z_k).
	$$ 
	In particular,  one has that  $(y, {\z})\in\bbQ $.
\end{proposition}
%
%
\begin{proof}
 Moving from Proposition  \ref{lem:2}, the proof proceeds
along the lines of the classical Modica-Mortola $\Gamma$-convergence
result \cite{MM}.  As  $ \liminf_{k\to+\infty}  \mathcal{F}^{\,\rm
  int}_{\eps_k}(y,\z)<\infty$  and  $\Phi(s)=0$ only for
$s=0,1$,  we have that  
	$
	{\z}=\chi_{F},
	$
	for some measurable set $F\subset\O^y $.
	By  using  the coarea formula  we deduce that 
	\begin{align*}
	\mathcal{F}^{\,\rm int}_{\eps_k}(y_k,\z_k)&=\int_{\Omega^{y_k}}\Big( \frac{\eps_k}2|\nabla 
	\z_k|^2+\frac 1{\eps_k} \Phi(\z_k)\Big)\,\d\xi\\
&\geq \int_{\Omega^{y_k}}\sqrt{2\Phi(\z_k)}\,|\nabla
	\z_k|\,\d\xi=\int_{0}^{1}\sqrt{2\Phi(s)}\,\per(\{\z_k>s\},\O^{y_k})\,\d s
	\end{align*}
	 Given any  $\delta\in (0,1)$ and
        $s\in[\delta,1-\delta]$  one has that 
	$$
	|\{\xi\in\O^{y_k}\ | \ \z_k>s\}\Delta F|\leq\frac{1}{\delta} \|\zeta_k-{\zeta}\|_{L^1(O^k)}+|\O^{y_k}\Delta \O^y|
	$$
	Therefore, by  applying Lemma \ref{lem:0} we get 
	$$
|\{\xi\in\O^{y_k} \ | \ \z_k>s\}\Delta F|\to 0\quad	\forall s\in[\delta,1-\delta].
	$$
	 Owing to Proposition \ref{lem:2} we obtain 
	$$
	\per(F,\O^y)\leq\liminf_{k\to+\infty}\per(
        \{\z_k>s\},\O^{y_k}) \quad 	\forall s\in[\delta,1-\delta].
	$$
	Hence,  as  $\delta\in (0,1)$,  by  applying
        the Fatou Lemma  one gets 
	\begin{align*}
	\liminf_{k\to+\infty} \mathcal{F}^{\,\rm int}_{\eps_k}(y_k,\z_k)&\geq \int_{0}^{1}\sqrt{2\Phi(s)}\,\liminf_{k\to+\infty} \per(\{\z_k>s\},\O^{y_k})\,\d s\\
	&\geq \int_{\delta}^{1-\delta}\sqrt{2\Phi(s)}\,\liminf_{k\to+\infty} \per(\{\z_k>s\},\O^{y_k})\,\d s\\
	&\geq \int_{\delta}^{1-\delta}\sqrt{2\Phi(s)}\,\cdot  \per(F,\O^{y})\,\d s
	\end{align*}
	and  the assertion follows as  $ \int_{\delta}^{1-\delta}\sqrt{2\Phi(s)}\,\d s\to \gamma $ for $\delta\to 0$.
\end{proof}
%
%
The existence of a recovery sequence is a direct consequence of the
classical Modica-Mortola theorem  \cite{MM}  as soon as  we assume that $\O^y$
is a Lipschitz domain.  Although this Lipschitz continuity could
fail to hold for general deformations, we can enforce it by asking
$y_0(\O)$ to be a Lipschitz domain where $y_0$ is the imposed boundary
deformation, see, e.g.,~\cite{Ball-1981} for a similar argument. Note
that the Lipschitz assumption on $\O^y$   was not needed  for
the $\Gamma$-$\liminf$ inequality of Proposition \ref{theo:liminf}. 
\begin{proposition}[Recovery sequence]\label{theo:limsup}
	
	If $(y,{\z})\in\overline\bbQ_{y_0}$, $y_0(\O)\subset\R^3$ 
        being  a Lipschitz domain, and  $F=\{\z=1\} $, there exists a sequence $\z_k\subset W^{1,2}(\O^y;[0,1])$ such that
	$$
	\lim_{k\to\infty}\|\z_k- {\z}\|_{L^1(\O^y)}=0 \quad\mbox{and}\quad \gamma\per( F, \O^y)+\mathcal{F}^{\,\rm bulk}(y,\z)=\lim_{k\to\infty}\mathcal{F}_{\eps_k}(y,\z_k).	$$
\end{proposition}
\begin{proof}	
	The sequence $\z_k$  is delivered  by the classical
        Modica-Mortola  construction  \cite{MM} applied to the
        functional $\mathcal{F}_\eps^{\,\rm int}(y,\z)$ with $y$
        fixed.  In fact, once the interface part convergence, the
        bulk part also follows because  $\mathcal{F}^{\,\rm bulk}(y,\z)$ is strongly continuous in $\z$.
\end{proof}

 We eventually combine the $\Gamma$-$\liminf$
inequality of Proposition \ref{theo:liminf} and the recovery-sequence
construction of Proposition
\ref{theo:limsup} in order to prove Theorem \ref{th2}. 

\begin{proofth2} Existence of minimizers $(y_k,\z_k)$ for  $ \FF_{\eps_k} $ has
  already been checked in  Proposition \ref{Prop:diffuse-ex}. 
 Let $(y_{0},\z_{0k})$ be the recovery sequence for $(y_0,\z_0)$
whose existence is proved in Proposition \ref{theo:limsup}. 
 By comparing with  $(y_{0},\z_{0k})$  one gets that 
\begin{align*}
  \FF^{\rm
  el} (y_k,\z_k)+\FF^{\rm
  int}_{\e_k}(y_k,\z_k)&=\FF_{\e_k}(y_k,\z_k)\leq \FF_{\e_k}
                         (y_{0},\z_{0k}) <C<\infty
\end{align*}
where we have used  the fact that  $\FF^{\rm
  int}_{\e_k}(y_{0},\z_{0k}) \to  \FF^{\rm
  int}_0(y_0,\z_0)$.  The latter bound  and the coercivity
\eqref{W-ass2} ensures that 
$y_k\to y$ weakly in $W^{1,p}(\O;\R^3)$ and 
 $|\O^y\Delta \O^{y_k}|\to 0$ by Lemma \ref{lem:0},  for some not
  relabeled subsequence.  On the other hand, since
 $\Omega^{y_k}$   contains any open set $A\subset\subset\O^y$ 
 for large $k$,  the  latter  bound on $\mathcal{F}^{\rm
   int}_{\eps_k}(y_k,\zeta_k)$ yields strong $L^1(A)$ compactness for
 the sequence $\zeta_k$. This implies the existence of $\zeta\in
 L^\infty(\Omega^y; [0,1])$ such that $\|\zeta_k-\zeta\|_{L^1(O^k)}\to
 0$  for some  not relabeled subsequence, as in the proof of
 Theorem \ref{th1}.   Proposition \ref{theo:liminf} ensures that
 $\zeta$ is a characteristic function and 
$$\FF^{\rm int}_0(y,\z) \leq \liminf_{k\to \infty} \FF^{\rm
  int}_{\e_k}(y_k,\z_k).$$
Moreover, for all $(\tilde y,\tilde \z) \in
\mathbb{Q}_{y_0}$ Proposition \ref{theo:limsup} ensures that there
exists a recovery sequence $\tilde \z_k$ such that $\FF_{\e_k}(\tilde y,\tilde
\z_k) \to \FF_0(\tilde y, \tilde \z)$.
As the bulk term $\FF^{\rm bulk}$ is lower semicontinuous, we conclude
that 
$$\FF_0(y,\z) \leq \liminf_{k\to \infty} \FF_{\e_k}(y_k,\z_k)\leq
\liminf_{k\to \infty} \FF_{\e_k}(\tilde y,\tilde \z_k) = \FF_0(\tilde y, \tilde \z).$$
Hence, $(y,\z)$ minimizes
$\FF_0$ on $\mathbb{Q}_{y_0}$. 
\end{proofth2}

%
%
\section*{Acknowledgements}  
This research is supported by the
FWF-GA\v{C}R project    I\,2375-16{-}34894L  and by the 
OeAD-M\v{S}MT  project   CZ~17/2016-7AMB16AT015.  M.K.~ was further supported by the GA\v{C}R project 18-03834S.
U.S.\ acknowledges the support by the Vienna Science and Technology Fund (WWTF)
through Project MA14-009 and by the Austrian Science Fund (FWF)
projects F\,65  and  P\,27052.  
The authors are indebted to 
L.~Ambrosio and M.~\v Silhav\'y for  inspiring conversations. 
%
%

%
%
\end{document}